\documentclass[12pt]{amsart}
\usepackage[margin=1in]{geometry} 
\usepackage{amsmath}
\usepackage{amsthm}
\usepackage{braket}
\usepackage[authoryear]{natbib}
\usepackage{hyperref}

\newtheorem{theorem}{Theorem}
\newtheorem{definition}{Definition}
\newtheorem{problem}{Problem}

\begin{document}

\title{On the Classical Hardness of Spoofing Linear Cross-Entropy Benchmarking}
\author{Scott Aaronson\textsuperscript{*} \and Sam Gunn\textsuperscript{\textdagger}}
\thanks{\textsuperscript{*} University of Texas at Austin. Email: \texttt{aaronson@cs.utexas.edu}. Supported by a Vannevar Bush Fellowship from the US Department of Defense, a Simons Investigator Award, and the Simons ``It from Qubit" collaboration.}
\thanks{\textsuperscript{\textdagger} University of Texas at Austin. Email: \texttt{samgunn111@utexas.edu}.}

\begin{abstract}
Recently, Google announced the first demonstration of quantum
computational supremacy with a programmable superconducting processor (\cite{Google}).
Their demonstration is based on collecting samples from the
output distribution of a noisy random quantum circuit, then
applying a statistical test to those samples called Linear
Cross-Entropy Benchmarking (Linear XEB). This raises a
theoretical question: How hard is it for a classical computer to
spoof the results of the Linear XEB test? In this short note, we adapt an
analysis of \cite{QUATH} to prove a conditional hardness
result for Linear XEB spoofing. Specifically, we show that the
problem is classically hard, assuming that there is no efficient
classical algorithm that, given a random $n$-qubit quantum circuit $C$,
estimates the probability of $C$ outputting a specific output string, say $0^n$, with
mean squared error even slightly better than that of the trivial estimator that
always estimates $1/2^n$. Our result automatically encompasses the case
of noisy circuits.
\end{abstract}

\maketitle
\section{Introduction}
Quantum computational supremacy refers to the solution of a well-defined computational task by a
programmable quantum computer in significantly less time than is required by the best known
algorithms running on existing classical computers, for reasons of asymptotic scaling. It is a
prerequisite for useful quantum computation, and is therefore seen as a major milestone in the field.
The task of sampling from random quantum circuits (called RCS) is one
proposal for achieving quantum supremacy (\cite{BISB}, \cite{BFNV}, \cite{QUATH}). Unlike
other proposals such as Boson Sampling (\cite{BosonSampling}) and Commuting Hamiltonians (\cite{IQP}), RCS involves a
universal quantum computer -- one theoretically capable of applying any unitary transformation. Achievement of quantum supremacy via RCS is therefore a
demonstration wherein a \emph{universal} quantum computer is able to outperform any existing classical
computer on a computational task.

A research team based at Google has announced a demonstration of quantum computational supremacy, by sampling the output distributions
of random quantum circuits (\cite{Google}). To verify that their circuits were working correctly, they
tested their samples using Linear Cross-Entropy Benchmarking (Linear XEB).
This test simply checks that the observed samples tend to concentrate on the outputs that
have higher probabilities under the ideal distribution for the given quantum circuit.
More formally, given samples $z_1, \hdots, z_k \in \{0,1\}^n$, Linear XEB entails checking that
$\mathbb{E}_i[P(z_i)]$ is greater than some threshold $b/2^n$, where $P(z)$ is the probability of observing $z$ under the ideal distribution. In the regime of 40-50 qubits, these probabilities can
be calculated by a classical supercomputer with enough time.

While there is some support for the conjecture that no classical algorithm can efficiently
sample from the output distribution of a random quantum circuit (\cite{BFNV}),
less is known about the hardness of directly spoofing a test like Linear XEB.
Results about the hardness of sampling are not quite results about the hardness of spoofing
Linear XEB; a device could score well on Linear XEB while being far from correct in total variation
distance by, for example, always outputting the items with the $k$ highest probabilities.

Under the assumption that the noise in the device is ``depolarizing" -- that a sample from the
circuit was sampled correctly with probability $b-1$ and otherwise sampled uniformly at random -- 
there is stronger evidence that it is difficult to spoof Linear XEB. Namely, if there is a classical
algorithm for sampling from a quantum circuit with depolarizing noise in time $T$, then with the help of an all-powerful but untrusted prover, one can calculate a good estimate for output probabilities in time $10T/(b-1)$ with high probability over circuits.
Together with results of \cite{SimBounds}, it follows that under the Strong Exponential Time
Hypothesis there is a quantum circuit from which one
cannot classically sample with depolarizing noise in time $(b-1) 2^{(1-o(1))n}$
(\cite{Google}, Supplement XI). We are not aware of any evidence that does not depend
on such a strong assumption about the noise.

However, \cite{QUATH} were able to prove the hardness of a different, related
verification procedure from a strong \emph{hardness} assumption they called the Quantum Threshold
Assumption (QUATH). Informally, QUATH states that it is impossible for a polynomial-time
classical algorithm
to guess whether a specific output string like $0^n$ has greater-than-median probability of
being observed
as the output of a given $n$-qubit quantum circuit, with success probability $1/2+\Omega(1/2^n)$.
They went on to investigate algorithms for breaking QUATH by
estimating the output amplitudes of quantum circuits. For certain classes of 
circuits, output amplitudes can be efficiently calculated, but in general even efficiently sampling
from the output distribution is impossible unless the polynomial hierarchy collapses
(\cite{BosonSampling}, \cite{IQP}).
\cite{QUATH} found an algorithm for calculating amplitudes of arbitrary circuits
that runs in time $d^{O(n)}$ and $poly(n,d)$ space, where $d$ is the circuit depth.
This is now used in some state-of-the-art simulations, but is still too slow and of the wrong
form to violate QUATH for larger circuits, as there is no way to trade the accuracy for polynomial-time efficiency.

Here, we formulate a slightly different assumption that we call XQUATH and
show that it implies the hardness of spoofing Linear XEB. Like QUATH, the 
new assumption is quite strong, but makes no reference to sampling. In particular,
while we don't know a reduction, refuting XQUATH seems essentially as hard as refuting QUATH.
Note that our result says nothing, one way or the other, about the possibility of improvements to
algorithms for calculating amplitudes. It just says that there's nothing particular to spoofing Linear
XEB that makes it easier than nontrivially estimating amplitudes.

Indeed, since the news of the Google group's success broke, at least three results
have potentially improved on the classical simulation efficiency, beyond what Google had considered.
First, \cite{Gray} was able to optimize tensor network contraction methods to
obtain a faster classical amplitude estimator, though it's not yet clear whether this will be competitive
for calculating millions of amplitudes at once.
Second, \cite{IBM} argued that, by using secondary storage, existing classical supercomputers
should be able to simulate the experiments done at Google in a few days. Third, \cite{2DCircuits} produced an efficient algorithm for approximately simulating average-case quantum circuits from a certain distribution of \emph{constant} depth, which is impossible to efficiently exactly simulate classically in the worst-case unless the polynomial hierarchy collapses.
% NEW^ ...... definitely needs to be stated differently...
This algorithm does not work for circuits as deep as those used by \cite{Google}.
Our result provides some explanation for why these improvements had to target the general problem
of amplitude estimation, rather than doing anything specific to the problem of spoofing Linear XEB.

\section{Preliminaries}
Throughout this note we will refer to random quantum circuits. Our results apply to circuits chosen
from any distribution $\mathcal{D}$ over circuits on $n$ qubits that is unaffected by appending
NOT gates to any subset of the qubits at the end of the circuit.\footnote{The experiment of \cite{Google} did not have this property. However, since applying NOT gates at the end of a circuit is equivalent to relabeling outputs, this issue is not physical.}
For every such distribution there is a corresponding version of XQUATH.
For instance, we could consider a
distribution where $d$ alternating layers of random single- and neighboring two-qubit gates are
applied to a square lattice of $n$ qubits, as in Google's experiment.
Our assumption XQUATH states that no efficient classical algorithm can estimate the
probability of such a random circuit $C$ outputting $0^n$, with mean squared error even slightly lower
than the trivial algorithm that always estimates $1/2^n$.
\begin{definition}[XQUATH, or Linear Cross-Entropy Quantum Threshold Assumption]
There is no polynomial-time classical algorithm that takes as input a quantum circuit $C \leftarrow \mathcal{D}$ and produces an estimate $p$ of $p_0 = \Pr[C \text{ outputs } 0^n]$ such that\footnote{The reason for the bound being $2^{-3n}$ will emerge from our analysis.}
$$\mathbb{E}[(p_0 - p)^2] = \mathbb{E}[(p_0 - 2^{-n})^2] - \Omega(2^{-3n})$$
where the expectations are taken over circuits $C$ as well as the algorithm's internal randomness.
\end{definition}
The simplest way to attempt to refute XQUATH might be to try $k$ random Feynman paths of the circuit, all of
which terminate at $0^n$, and take the empirical mean over their contributions to
the amplitude. However, this approach will only yield an improvement in mean squared error over the trivial
algorithm that decays exponentially with the number of {\it gates} in the circuit, rather than the
number of qubits. When the number of gates is much larger than $3n$, as in \cite{Google}, it is
clear that this approach cannot violate XQUATH.
As mentioned above, even the best existing quantum simulation algorithms do not
appear to significantly help in refuting XQUATH for reasonable circuit distributions.

The problem XHOG is to generate outputs of a given quantum circuit that have high expected
squared-magnitude amplitude. These outputs are required to be distinct for reasons that will become clear in the proof of Theorem 1.
\begin{problem}[XHOG, or Linear Cross-Entropy Heavy Output Generation]
Given a circuit $C$, generate $k$ distinct samples $z_1, \hdots, z_k$ such that $\mathbb{E}_i[\left|\bra{z_i} C \ket{0^n}\right|^2] \ge b/2^n$.
\end{problem}
The interesting case is when $b > 1$, and we will generally think of $b$ as a constant. Without fault-tolerance, $b-1$ will quickly become very small for circuits larger than the experiment can
handle. This is a difficulty of applying complexity theory to finite experiments, which fail
when the problem instance is too large.

When the depth is large enough, the output probabilities $p$ of almost all circuits are empirically
observed to be accurately described by the distribution $2^n e^{-2^n p}$, although
this has only been rigorously proven in some special cases (\cite{BISB}, \cite{Google}, \cite{TDesigns}). Under this
assumption, for observed outputs $z$ from ideal circuits $C \leftarrow \mathcal{D}$ we have
$$\mathbb{E}[\left|\bra{z} C \ket{0^n}\right|^2] \approx \int_0^\infty \frac{x}{2^{n}} x e^{-x} dx = \frac{2}{2^n}$$
So we expect an ideal circuit to solve XHOG with $b \approx 2$, and a noisy circuit to solve
XHOG with $b$ slightly larger than 1. Theorem 1 says that, assuming
XQUATH, solving XHOG with $b>1$ is hard to do classically with many samples and high probability.
For completeness, we show in the Appendix that with Google's number of samples and
estimated circuit fidelity, they would be expected to solve XHOG with sufficiently high probability.

\section{Proof of the Reduction}
We now provide a reduction from the problem in XQUATH to XHOG. Since we only call the XHOG
algorithm once in the reduction, and all other steps are efficient, solving XHOG actually requires as
many computational steps as solving the problem in XQUATH, minus $O(k)$.
\begin{theorem}
Assuming XQUATH, no polynomial-time classical algorithm can solve XHOG with probability $s > \frac{1}{2} + \frac{1}{2b}$, and
$$k \ge \frac{1}{((2s-1)b-1)(b-1)}.$$
With $b = 1+\delta$ and $s = \frac{1}{2} + \frac{1}{2b} + \epsilon$, the right-hand side is approximately $1/2\epsilon\delta$.
\end{theorem}
\begin{proof}
Suppose that $A$ is a classical algorithm solving XHOG with the parameters above. Given a quantum circuit $C \leftarrow \mathcal{D}$, first draw a uniformly random $z \in \{0,1\}^n$, and apply NOT gates at the end of $C$ on qubits $i$ where $z_i = 1$ to get a circuit $C'$. According to our assumption on $\mathcal{D}$, $C'$ is distributed exactly the same as $C$, even conditioned on a particular $z$. Also, $\bra{0^n} C \ket{0^n} = \bra{z} C' \ket{0^n}$, so $\Pr[C \text{ outputs } 0^n] = \Pr[C' \text{ outputs } z]$. Call this probability $p_0$.

Run $A$ on input $C'$ to get $z_1, \hdots, z_k$ with $\mathbb{E}_i[\left|\bra{z_i} C \ket{0^n}\right|^2] \ge b2^{-n}$ (when $A$ succeeds). If $z\in\{z_i\}$, then our algorithm outputs $p = b 2^{-n}$; otherwise it outputs $p = 2^{-n}$.

Let $X = (p_0 - 2^{-n})^2 - (p_0 - p)^2$. Then
\begin{align*}
\mathbb{E}[X \mid z\in\{z_i\} \text{ and } A \text{ succeeded}] &= 2 \cdot 2^{-n}(b - 1) \cdot \mathbb{E}[p_0 \mid z\in \{z_i\} \text{ and } A \text{ succeeded}] + 2^{-2n} (1- b^2) \\
&\ge 2 \cdot 2^{-n}(b - 1) (b2^{-n}) + 2^{-2n} (1- b^2) \\
&= 2^{-2n} (b-1)^2 \\
\mathbb{E}[X \mid z\in\{z_i\} \text{ and } A \text{ failed}] &= 2 \cdot 2^{-n}(b - 1) \cdot \mathbb{E}[p_0 \mid z\in\{z_i\} \text{ and } A \text{ failed}] + 2^{-2n} (1- b^2) \\
&\ge -2^{-2n} (b^2-1) \\
\intertext{Since $\mathbb{E}[X \mid z\not\in\{z_i\}] = 0$, and since $z$ is uniformly random even conditioned on the output of $A$ and its success or failure,}
\mathbb{E}[X] &= 2^{-n} k s \cdot \mathbb{E}[X \mid z\in\{z_i\} \text{ and } A \text{ succeeded}] \\
&\ \ \ \ + 2^{-n} k (1-s) \cdot \mathbb{E}[X \mid z\in\{z_i\} \text{ and } A \text{ failed}] \\
&\ge 2^{-3n} k ((2s-1)b -1) (b-1)
\end{align*}
which is $\Omega(2^{-3n})$ as long as $k \ge 1/((2s-1)b-1)(b-1)$. This completes the proof.
\end{proof}
One simple instance of the theorem is to take $s = \frac{3}{4} + \frac{1}{4b}$ and $k = 2(b-1)^{-2}$.
Note that even with $s=1$, we need $k \ge (b-1)^{-2}$ samples for the proof to work. In fact, if
the number of samples $k$ is much smaller than $(b-1)^{-2}$, then even sampling uniformly at
random would pass XHOG with non-negligible probability (see Appendix).

\section{Open Problems}
We conclude with two open problems related to our reduction.
\begin{itemize}
\item Can the classical hardness of spoofing Linear XEB be based on a more secure assumption?
Is there a similar assumption to XQUATH that is {\it equivalent} to the classical hardness of XHOG?
\item Is XQUATH true? What is the relationship of XQUATH to QUATH?
\end{itemize}

\section*{Acknowledgements}
\noindent We thank Umesh Vazirani, Boaz Barak, Daniel Kane, Ryan O'Donnell, and Patrick Rall for
helpful discussions on the subject of this note.

\bibliographystyle{plainnat}
\bibliography{references.bib}

\section*{Appendix}
\subsection*{Probability of Failing XHOG}
We would like to be confident that one is solving XHOG with sufficiently high probability $s$ for Theorem 1 to apply, without having to perform the experiment enough times to verify this directly. If one has an estimate for their circuit fidelity and assumes depolarizing noise, the following calculation shows that this is possible. For clarity, we will restrict to the case of solving XHOG with $s = 3/4+1/4b$ and $k = 4(b-1)^{-2}$.

Let $Y = \left|\bra{z}C\ket{0^n}\right|^2$, where $z$ is sampled from our XHOG device which was given $C \leftarrow \mathcal{D}$. We show that if one can produce samples with estimated fidelity $\mathbb{E}[Y] \ge (2b-1)/2^n$, they can solve XHOG with probability $s = 3/4+1/4b$.

If our device is noise-free,
$$\text{Var}(Y) \approx \int_0^\infty x e^{-x} \left(\frac{x}{2^n} - \frac{2}{2^n}\right)^2 dx = \frac{2}{2^{2n}}$$
and if we are sampling from the uniform distribution,
$$\text{Var}(Y) \approx \int_0^\infty e^{-x} \left(\frac{x}{2^n} - \frac{1}{2^n}\right)^2 dx = \frac{1}{2^{2n}}$$
By the law of total variance, for a noisy device sampling from the ideal distribution with
probability $p$ and the uniform distribution with probability $1-p$,
$$\text{Var}(Y) \approx \frac{2p}{2^{2n}} + \frac{1-p}{2^{2n}} + \frac{p (2 - p - 1)^2}{2^{2n}} + \frac{(1-p) (1 - p - 1)^2}{2^{2n}} = \frac{1+2p-p^2}{2^{2n}} \le \frac{2}{2^{2n}}$$
So for a noisy device, the standard deviation is at most $\sqrt{2} / 2^n$. Let $\bar{Y}_k$ be the average of $Y$ over $k$ samples.
Then the standard deviation of $\bar{Y}_k$ is at most $\sqrt{2} / k 2^n$.

Applying Chebyshev's inequality and the fact that $k \ge 4/(b-1)^2 \ge 4/(b-1)^{3/2}$,
$$\Pr[\bar{Y}_k \le b/2^n] \le \Pr[\left|\bar{Y}_k - \mathbb{E}[Y]\right| \ge (b-1)/2^n] \le \frac{2}{(b-1)^2 k^2} \le (b-1)/8$$
In particular, the success probability is $s \ge 1-(b-1)/8 \ge 1-(b-1)/4b = 3/4+1/4b$.

The Google team estimated that they were sampling with
$\mathbb{E}[Y] \approx 1.002/2^n$ for their largest circuits, so to solve XHOG with $b = 1.001$ and $s = 3/4+1/4b$ they needed 4 million samples, much less than the 30 million they did take.% For this calculation we used the assumption that the noise is depolarizing, but all that was needed was to show that $\text{Var}(Y)$ is sufficiently small.

%%%%%%%%%%%%% SCRATCH %%%%%%%%%%%%%
%$$0 < Y < n/2^n$$
%The worst sampler for us just gives $Y=n/2^n$ w.p. $p$, and like $1/2^n$ otherwise. So
%$$p n + (1-p)b \ge 2b-1$$
%$$p \ge \frac{b-1}{n-b}$$
%So with only like $n/(b-1)$ samples we should be ok.
%%%%%%%%%%%%% SCRATCH %%%%%%%%%%%%%

\subsection*{Uniformly-Random Sampling}
In this section we show that one needs to take about $(b-1)^{-2}$ samples before the values $\left|\bra{z_i} C \ket{0^n}\right|^2$ can distinguish between the cases where the $z_i$ are drawn uniformly at random or passing XHOG. The KL divergence of taking a single sample is
$$\text{KL}(e^{-p}, p e^{-p}) \approx \int_0^\infty e^{-p} (b (p-1)-p+2) \log (b (p-1)-p+2) dp$$
It is not hard to calculate that the Taylor expansion of the above around $b=1$ is $(b-1)^2/2 + O((b-1)^3)$. By additivity, the KL divergence for $k$ samples is approximately $k (b-1)^2/2$. By Pinsker's inequality, the total variation distance is at most
$$\sqrt{k (b-1)^2/4}$$
Therefore, in order to have total variation distance independent of $b$, one needs $k \approx (b-1)^{-2}$.

%%%%%%%%%%%%% SCRATCH %%%%%%%%%%%%%
%\begin{align*}
%&\left.\frac{d}{db}\right|_{b=1} \int_0^\infty e^{-p} (b (p-1)-p+2) \log (b (p-1)-p+2) dp \\
%&= 0 + \int_0^\infty e^{-p} (p-1) dp \\
%&= 0 \\
%&\left.\frac{d^2}{db^2}\right|_{b=1} \int_0^\infty e^{-p} (b (p-1)-p+2) \log (b (p-1)-p+2) dp \\
%&= \left.\frac{d}{db}\right|_{b=1} \int_0^\infty e^{-p} (p-1) \log (b (p-1)-p+2) dp \\
%&= \int_0^\infty e^{-p} (p-1)^2 dp \\
%&= 1
%\end{align*}
%%%%%%%%%%%%% SCRATCH %%%%%%%%%%%%%

\end{document}